\documentclass[a4paper,11pt]{article}
\usepackage{algorithm}
\usepackage{algorithmic}
\usepackage{mathrsfs}
\usepackage{amsmath,amsthm,amssymb}
\usepackage{cite}
\usepackage{txfonts}
\usepackage{graphics}
\usepackage{float}
\usepackage{epstopdf}
\usepackage{epsfig}
\usepackage{tikz}
\usepackage{caption}

\usepackage[T1]{fontenc}
\usepackage[utf8]{inputenc}
\usepackage{authblk}

\newtheorem{definition}{Definition}

\newtheorem{theorem}{Theorem}

\newtheorem{lemma}{Lemma}

\title{A Tighter Relation between Sensitivity and Certificate Complexity}
\author[1]{Kun He}
\author[1]{Qian Li}
\author[1]{Xiaoming Sun}
\affil[1]{Institute of Computing Technology, Chinese Academy of Sciences, Beijing, China\\

\texttt{\{hekun, liqian, sunxiaoming\}@ict.ac.cn}}
\date{\vspace{0em}}
\begin{document}
\maketitle
\begin{abstract}
The \emph{sensitivity conjecture} which claims that the sensitivity complexity is polynomially related to block sensitivity complexity, is one of the most important and challenging problem in decision tree complexity theory. 
Despite of a lot of efforts, the best known upper bound of block sensitivity, as well as the certificate complexity, are still exponential in terms of sensitivity: $bs(f)\leq C(f)\leq\max\{2^{s(f)-1}(s(f)-\frac{1}{3}),s(f)\}$~\cite{APV16}. In this paper, we give a better upper bound of $bs(f)\leq C(f)\leq(\frac{8}{9} + o(1))s(f)2^{s(f) - 1}$. The proof is based on a deep investigation on the structure of the {\em sensitivity graph}. We also provide a tighter relationship between $C_0(f)$ and $s_0(f)$  for functions with $s_1(f)=2$.
\end{abstract}
\noindent

\section{Introduction}
The relation between sensitivity complexity and other decision tree complexity measures is one of the most important topic in Boolean function complexity theory. Sensitivity complexity is first introduced by Cook, Dwork and Reischuk~\cite{Cook,CDR86} to study the time complexity of CREW-PRAMs.
Nisan~\cite{nisan1991crew} then introduced the concept of block sensitivity, and demonstrated the remarkable fact that block sensitivity can fully characterize the time complexity of CREW-PRAMs.
Block sensitivity turns out to be polynomially related to a number of other complexity measures for Boolean functions~\cite{buhrman}, such as decision tree complexity, certificate complexity, polynomial degree and quantum query complexity, etc. One exception is sensitivity. So far it is still not clear whether sensitivity complexity could be exponentially smaller than block sensitivity and other measures. The famous sensitivity conjecture, proposed by Nisan and Szegedy in 1994~\cite{NS94}, asserts that block sensitivity and sensitivity complexity are also polynomially related. According to the definition of sensitivity and block sensitivity, it is easy to see that $s(f)\leq bs(f)$ for any total Boolean function $f$. But in the other direction, it is much harder to prove an upper bound of block sensitivity in terms of sensitivity complexity.
Despite of a lot of efforts, the best known upper bound of block sensitivity is still exponential in terms of sensitivity: $bs(f)\leq C(f)\leq\max\{2^{s(f)-1}(s(f)-\frac{1}{3}),s(f)\}$ ~\cite{APV16}. The best known separation between sensitivity and block sensitivity complexity is quadratic~\cite{AS11}: there exist a sequence of Boolean functions $f$ with $bs(f)=\frac{2}{3}s(f)^2-\frac{1}{3}s(f)$. Recently, Tal~\cite{Tal16} showed that any upper bound of the form $bs_l(f)\leq s(f)^{l-\varepsilon}$ for $\varepsilon>0$ implies a subexponential upper bound on $bs(f)$ in terms of $s(f)$.  Here $bs_l(f)$, the $l$-block sensitivity, defined by Kenyon and Kutin~\cite{KK04}, is the block sensitivity with the size of each block at most $l$. Note that the sensitivity conjecture is equivalent to ask whether sensitivity complexity is polynomially related to certificate complexity, decision tree complexity, Fourier degree or any other complexity measure which is polynomially related to block sensitivity. Ben-David~\cite{Sha16} provided a cubic separation between quantum query complexity and sensitivity, as well as a power 2.1 separation between certificate complexity and sensitivity.
While to solve the sensitivity conjecture seems very challenging for general Boolean functions, special classes of functions have also been investigated, such as functions with graph properties~\cite{Turan}, cyclically invariant functions\cite{Chakraborty}, small alternating number~\cite{LZ16}, constant depth regular read-$k$ formulas \cite{MFCS16},  etc\cite{ST16}. We recommend readers \cite{HKP11} for an excellent survey about the sensitivity conjecture. For other recent progresses, see \cite{Bop12,AP14,ABG+14,AV15,GKS15,Sze15,GNS+16,GSTW16}.\\

\noindent\textbf{Our Results.}
 In this paper, we give a better upper bound of block sensitivity in terms of sensitivity.
\begin{theorem}\label{theorem of main result}For any total Boolean function $f:\{0,1\}^n\rightarrow\{0,1\}$,
$$bs(f)\leq C(f)\leq(\frac{8}{9} + o(1))s(f)2^{s(f) - 1}.$$
Here $o(1)$ denotes a term that vanishes as $s(f) \rightarrow \infty$.
\end{theorem}

Ambainis et al.~\cite{APV16} also investigated the function with $s_1(f)=2$, and showed that $C_0(f)\leq \frac{9}{5}s_0(f)$ for any Boolean function with $s_1(f)=2$. In this paper, we also improve this bound,
\begin{theorem}\label{theorem of small s_1(f)}Let $f$ be a Boolean function with $s_1 (f)=2$,
$$C_0(f) \leq \frac {37 + \sqrt 5}{22} s_0(f)\approx 1.783 s_0(f).$$
\end{theorem}

 \noindent\textbf{Organization.} We present preliminaries in Section~\ref{section:preliminaries}. We give the overall structure of our proof for Theorem \ref{theorem of main result} in Section~\ref{section:The overall structure of the proof} and the detailed proofs for lemmas in
 Section~\ref{section:Proofs of the lemmas}. We prove Theorem \ref{theorem of small s_1(f)} in Section~\ref{section:Proofs of theorem of small s_1(f)}. Finally, we conclude this paper in Section~\ref{section:conclusion}.

\section{Preliminaries}\label{section:preliminaries}
Let $f:\{0,1\}^n\rightarrow\{0,1\}$ be a Boolean function. For an input $x\in\{0,1\}^n$ and a subset $B \subseteq [n]$, $x^B$ denotes the input obtained by flipping all the bit $x_j$ such that $j\in B$.
\begin{definition}
The $\emph{sensitivity}$ of $f$ on input $x$ is defined as $s(f,x):=|\{i: f(x)\neq f(x^i)\}|$. The sensitivity $s(f)$ of $f$ is defined as $s(f):=\max_xs(f,x)$. The $b$-sensitivity $s_b(f)$ of $f$, where $b\in \{0,1\}$, is defined as $s_b(f)=\max_{x\in f^{-1}(b)}s(f,x)$.
\end{definition}

\begin{definition}
The $\emph{block sensitivity}$  $bs(f,x)$ of $f$ on input $x$ is the maximum number of disjoint subsets $B_1,B_2,\cdots,B_r$ of $[n]$ such that for all $j \in [r]$, $f(x)\neq f(x^{B_j})$. The block sensitivity of $f$ is defined as $bs(f)=\max_x bs(f,x)$. The $b$-block sensitivity $bs_b(f)$ of $f$, where $b\in \{0,1\}$, is defined as $bs_b(f)=\max_{x\in f^{-1}(b)}bs(f,x)$.
\end{definition}
\begin{definition}
A partial assignment is a function $p:\{0,1\}^n\rightarrow\{0,1,\ast\}$. We call $S=\{i|p(i)\neq\ast\}$ the support of this partial assignment. We define the co-dimension of $p$ denoted by co-dim($p$) to be $|S|$. We say $x$ is consistent with $p$ if $x_i=p_i$ for every $i\in S$. $p$ is called a $b$-certificate if $f(x)=b$ for any $x$ consistent with $p$, where $b\in \{0,1\}$. For $B\subseteq S$, $p^B$ denotes the partial assignment obtained by flipping all the bit $p_j$ such that $j\in B$. For $i\in [n]/S$, $p_{i=0}$ denotes the partial assignment obtained by setting $p_{i}=0$.\footnote{The function $p$ can be viewed as a vector, and we sometimes use $p_i$ to represent $p(i)$.}
\end{definition}


\begin{definition}
The $\emph{certificate complexity}$ $C(f, x)$ of $f$ on $x$ is the minimum co-dimension of $f(x)$-certificate that $x$ is consistent with. The certificate complexity $C(f)$ of $f$ is defined as $C(f)=\max_{x} C(f,x)$. The $b$-certificate complexity $C_b(f)$ of $f$ is defined as $C_b(f)=\max_{x\in f^{-1}(b)} C(f,x)$
\end{definition}

In this work we regard $\{0,1\}^{n}$ as a set of vertices for a \emph{$n$-dimensional hypercube $Q_n$}, where two nodes $x$ and $y$ has an edge if and only if the Hamming distance between them is 1. A Boolean function $f:\{0,1\}^n\rightarrow\{0,1\}$ can be regarded as a 2-coloring of the vertices of $Q_n$, where $x$ is \emph{black} if $f(x)=1$ and  $x$ is \emph{white} if $f(x)=0$. Let $f^{-1}(1)=\{x|f(x)=1\}$ be the set of all black vertices. If $f(x)\neq f(y)$, we call the edge $(x,y)$ a sensitive edge and $x$ is sensitive to $y$ ($y$ is also sensitive to $x$). We regard a subset $S\in\{0,1\}^{n}$ as the subgraph $G$ induced by the vertices in $S$. Define the size of $G$, $|G|$, as the size of $S$. It is easy to see that $s(f,x)$ is the number of neighbors of $x$ which has the different color with $x$. It is easy to see that a certificate is a  monochromatic subcube, and $C(f,x)$ is the  minimum co-dimension of  monochromatic subcube which contains $x$.

There is a natural bijection between  the partial assignments and the subcubes, where a partial assignment $p$ corresponds to a subcube induced the vertices consistent with $p$. Without ambiguity, we sometimes abuse these two concepts.

\begin{definition}
Let $G$ and $H$ be two induced subgraphs of $Q_n$. Let $G \cap H$ denote the graph induced on $V(G) \cap V(H)$. For any two subcubes $G$ and $H$, we call $H$ a neighbor cube of $G$ if their corresponding partial assignments $p_G$ and $p_H$ satisfying $p_G = p_H^{i}$ for some $i$.
\end{definition}

Our proofs rely on the following result proved by Ambainis and Vihrovs \cite{AV15}.
\begin{lemma}\label{lemma of Ambainis}~\cite{AV15} Let $G$ be a non-empty induced subgraph of $Q_k$ satisfying that the sensitivity of every vertices in $G$ is at
most s, then either $|G| \geq \frac{3}{2} \cdot 2^{k-s}$, or $G$ is a subcube of $Q_k$ with co-dim($G$)=$s$ and $|G|=2^{k-s}$.
\end{lemma}
\begin{definition}
For any set $S\subseteq\{0,1\}^n$, let $s(f,S)$ to be $\sum_{x \in S}s(f,x)$, the average sensitivity of $S$ is defined by $s(f,S)/|S|$.
\end{definition}

\section{The Sketch of Proof}\label{section:The overall structure of the proof}

In this section, we give the sketch of the proof of Theorem \ref{theorem of main result}. We first present some notations used in the proof. Let $f$ be an $n$-input Boolean function.
Let $z$ be the input with $f(z)=0$ and $C(f,z)=C_0 (f)=m$, W.l.o.g, we assume $z=0^n$, then there exists a 0-certificate of co-dim $C_0(f)$ consistent with $z$, and  let $H$ be the one with maximum average sensitivity if there are many such 0-certificates. W.l.o.g, assume $H=0^m\ast^{n-m}$. Among the $m$ neighbor cubes of $H$, from Lemma \ref{lemma of Ambainis} we have either $|H^i\cap f^{-1}(1)|\geq\frac{3}{2}\cdot \frac{|H|}{2^{s_1(f)-1}}$ or $H^i\cap f^{-1}(1)$ is a co-dimensional $(s_1(f)-1)$  subcube of $H^i$ of size $\frac{|H|}{2^{s_1(f)-1}}$, which are called {\em heavy cube} and {\em light cube}, respectively. W.l.o.g, assume $H^1,H^2\cdots,H^l$ are light cubes and $H^{l+1},\cdots,H^{m}$ are heavy cubes, where $l\leq m$ is the number of light cubes.  For $k>m$, let $N_k^0=\{i\in[l]|(H^i\cap f^{-1}(1))_k=0\}$. Similarly, let $N_k^1=\{i\in[l]|(H^i\cap f^{-1}(1))_k=1\}$ and $N_k=N_k^0\cup N_k^1$.

For any  subcube $H'\subseteq H$, we use $s_l(f,H')$ ($s_h(f,H')$ respectively) to denote the number of sensitive edges of $H'$ adjacent to the light cubes (heavy cubes respectively). Similarly, for subcube $H'\subseteq H^i$ where $i\leq l$, we use $s_l(f,H')$ ($s_h(f,H')$ respectively) to denote the number of sensitive edges of $H'$ adjacent to $H^{1,i},\cdots, H^{i-1,i},H^{i},H^{i+1,i},\cdots,H^{l,i}$ ($H^{l+1,i},\cdots,H^{m,i}$ respectively). It is easy to see $s_l(f,H')+s_h(f,H')=s(f,H')$.

The main idea is show that there are many 1-inputs in the heavy cubes. To see why it works, consider the extremal case where there are no light cubes (i.e. $l=0$), then the average sensitivity of $H$ is no less than  $m\cdot\frac{3}{2^{s_1(f)}}$. Because the average sensitivity of $H$ can not exceed $s_0(f)$, we have $m\cdot\frac{3}{2^{s_1(f)}} \leq s_0(f)$ and $m\leq\frac{2}{3}s_0(f)2^{s_1(f)-1}$.

More specifically, the average sensitivity of $H$ is no less than $\frac{l}{2^{s_1(f)-1}}+\frac{3(m-l)}{2^{s_1(f)}}$. Let $L=s_0(f)2^{s_1(f)-1}/l$. If $L \geq 2$, we have $l\leq s_0(f)2^{s_1(f)-2}$ and $m\leq \frac{5}{6}s_0(f)2^{s_1(f)-1}$. In the following paper, we assume $L<2$. If $s_1(f)=1$, it has already been shown that $C_0(f)\leq s_0(f)$ \cite{APV16}. So we assume $s_1(f)\geq 2$ here. Hence, from $L<2$ we have $l>s_0(f)$. Note that if $i\in N^1_k$, then $H_{k=0}$ together with $H_{k=0}^{i}$ is another certificate of $z$ of the same co-dimension with $H$, thus according to the assumption that $H$ is the one with maximum average sensitivity, we have
\begin{center}
$s(f,H)-\left(s(f,H_{k=0})+s(f,H_{k=0}^i)\right)=s(f,H_{k=1})-s(f,H_{k=0}^i)\geq 0.$
\end{center}
  By summing over different cubes and different bits, we get
\begin{equation} \begin{aligned} \label{Main idea 2}
& \quad \sum_{k:|N_k^1| \geq s_1(f)-1} \sum_{i\in S_k^1} \left(s_h(f,H_{k=1}) - s_h(f,H_{k=0}^i)\right)
\\& = \sum_{k:|N_k^1| \geq s_1(f)-1} \sum_{i\in S_k^1} \left[ (s_l(f,H_{k=0}^{i}) -  s_l(f,H_{k=1})) - (s(f,H_{k=0}^i) - s(f,H_{k=1}))\right]
\\& \geq \sum_{k:|N_k^1| \geq s_1(f)-1} \sum_{i\in S_k^1}  \left(s_l(f,H_{k=0}^i) -  s_l(f,H_{k=1})\right)
\geq\frac{(s_1(f)-1)|H|}{2^{s_1(f)-1}}\sum_{k:|N_k^1| \geq s_1(f)-1}|N_k^0|
\\& \geq (\frac{1}{2}-o(1))\frac{(s_1(f)-1)^2|H|l}{2^{s_1(f)-1}}.
\end{aligned} \end{equation}
Here $o(1)$ denotes a term that vanishes as $s_1(f) \rightarrow \infty$, and  $S_k^1$ is a subset of $N_k^1$ of size $s_1(f)-1$. The second last inequality is due to the following lemma.
\begin{lemma}\label{lemma of light cubes}
$s_l(f,H_{k=0}^i)-s_l(f,H_{k=1})\geq \frac{|N_k^0|\cdot|H|}{2^{s_1(f) - 1}}$, for any $i\in N^1_k$.
\end{lemma}
The last inequality is due to
\begin{lemma}\label{corollary of N_0}
If $L<2$, then $\sum_{k:|N_k^1| \geq s_1(f)-1} (\frac{1}{2}-o(1))|N_k^0|\geq l(s_1(f)-1)$.
\end{lemma}
On the other side, we can show that
\begin{lemma}\label{lemma of heavy cubes}
$$\sum_{k:|N_k^1| \geq s_1(f)-1} \sum_{i\in S_k^1} \left(s_h(f,H_{k=1}) - s_h(f,H_{k=0}^i)\right)\leq (s_1(f) - 1)^2\sum_{l< t \leq m} |H^t\cap f^{-1}(1)|.$$
\end{lemma}
The proofs of these three lemmas are postponed to the next section. We first finish the proof of Theorem \ref{theorem of main result} here. Equality~\ref{Main idea 2} together with Lemma~\ref{lemma of heavy cubes} states that there are many 1-inputs in the heavy cubes, i.e.
$$\sum_{l<t\leq m}|H^t\cap f^{-1}(1)|\geq\frac{((\frac{1}{2}-o(1))l|H|}{2^{s_1(f) - 1}}.$$
Combined it with $$\frac{l}{2^{s_1(f)-1}}+\sum_{l< t \leq m} \frac{|H^t\cap f^{-1}(1)|}{|H|}\leq s_0(f),$$ we get $$l\leq(\frac{2}{3}+o(1))2^{s_1(f)-1}s_0(f).$$ Moreover, recall that $|H^t\cap f^{-1}(1)|/|H|\geq \frac{3}{2^{s_1(f)}}$, thus $$\frac{l}{2^{s_1(f)-1}}+\frac{3}{2}\cdot\frac{m-l}{2^{s_1(f)-1}}\leq s_0(f).$$ Therefore,
\begin{center}
$C_0(f)=m\leq (\frac{8}{9} + o(1))s_0(f)2^{s_1(f) - 1}.$
\end{center}
Here, $o(1)$ denotes a term that vanishes as $s_1(f) \rightarrow \infty$. Similarly, we can also obtain
\begin{center}
$C_1(f)\leq (\frac{8}{9} + o(1))s_1(f)2^{s_0(f) - 1}.$
\end{center}
Therefore,
\begin{center}
$C(f)\leq (\frac{8}{9} + o(1))s(f)2^{s(f)-1}.$
\end{center}
where $o(1)$ denotes a term that vanishes as $s(f) \rightarrow \infty$. 

\section{Proofs of the Lemmas}\label{section:Proofs of the lemmas}

\subsection{Proof of Lemma \ref{lemma of light cubes}}\label{subsection:Proofs of lemma of light cubes}
Before giving the proof of Lemma~\ref{lemma of light cubes}, we first state the following lemma which will be used.
\begin{lemma}\label{Congruent Lemma}If $i,j\in N_k$,
then $|H^{i,j}_{k=1} \cap f^{-1}(1)| \geq \frac{|H|}{2^{s_1(f) - 1}}$ and $|H^{i,j}_{k=0} \cap f^{-1}(1)|\geq \frac{|H|}{2^{s_1(f) - 1}}$.
\end{lemma}
\begin{proof} W.l.o.g, assume $i\in N_k^1$. For any $x \in H^i\cap f^{-1}(1)\subseteq H^i_{k=1}$, there are $(s_1(f)-1)$ vertices in $H^i$ as well as $x^i\in H$ sensitive to $x$, thus $x^{j}\in H^{i,j}$ is in $f^{-1}(1)$, since otherwise $x$ would be sensitive to $s_1(f) + 1$ vertices. Therefore, $|H^{i,j}_{k=1} \cap f^{-1}(1)|\geq |H^i\cap f^{-1}(1)|=\frac{|H|}{2^{s_1(f) - 1}}$. Similarly, if $j\in N_k^0$, we have $|H^{i,j}_{k=0} \cap f^{-1}(1)|\geq \frac{|H|}{2^{s_1(f) - 1}}$.

If $j\in N_k^1$, note that $H^{i,j}_{k=0}\cap f^{-1}(1)\neq \emptyset $, since otherwise $H^{i,j}_{k=0},H^{i}_{k=0},H^{j}_{k=0}$ and $H_{k=0}$ would become a larger monochromatic subcube containg  $z$, which is contradicted with the assumption of $H$. For any $y \in H^{i,j}_{k=0}\cap f^{-1}(1)$, $y$ is sensitive to $y^{i}\in H^i$ and $y^{j}\in H^j$, thus $y$ has at most $s_1(f) - 2$ sensitive edges in $H^{i,j}_{k=0}$. Therefore, $|H^{i,j}_{k=0}\cap f^{-1}(1)|\geq \frac{|H^{i,j}_{k=0}|}{2^{s_1(f) - 2}}=\frac{|H|}{2^{s_1(f) - 1}}$ according to Lemma~\ref{lemma of Ambainis}.
\end{proof}
\begin{proof} ({\bf Proof of Lemma~\ref{lemma of light cubes}}) Since $H_{k=1}\cap f^{-1}(1)=\emptyset$ and $H_{k=0}^i\cap f^{-1}(1)=\emptyset$, it is easy to see
$$s_l(f,H_{k=1})=\sum_{j=1}^l |H^j_{k=1}\cap f^{-1}(1)|=\frac{|N^1_k|\cdot|H|}{2^{s_1(f)-1}}+\frac{(l-|N_k|)|H|}{2^{s_1(f)}}=\frac{(l+|N^1_k|-|N^0_k|)|H|}{2^{s_1(f)}}.$$
Similarly,
\begin{center}
$s_l(f,H_{k=0}^i)=\sum_{j=1,j\neq i}^l |H^{i,j}_{k=0}\cap f^{-1}(1)|+|H^{i}\cap f^{-1}(1)|$.
\end{center}
If $j\notin N_k$, then for any $x\in H^{j}\cap f^{-1}(1)$, we have $x^{i}\in f^{-1}(1)$ since otherwise $x$ would have $s_1(f)+1$ sensitivity edges, thus $|H^{i,j}_{k=0}\cap f^{-1}(1)|\geq \frac{1}{2}|H^{i,j}_{k=0}\cap f^{-1}(1)|=\frac{|H|}{2^{s_1(f)}}$. If $j\in N_k$, $|H^{i,j}_{k=0}\cap f^{-1}(1)|\geq \frac{|H|}{2^{s_1(f)-1}}$ according to Lemma~\ref{Congruent Lemma}. Therefore, $s_l(f,H_{k=0}^i)\geq \frac{(l+|N^1_k|+|N^0_k|)|H|}{2^{s_1(f)}}= s_l(f,H_{k=1})+\frac{|N^0_k|\cdot|H|}{2^{s_1(f)-1}}$.
\end{proof}

\subsection{Proof of Lemma \ref{corollary of N_0}}\label{subsection:Proofs of Corollary of N_0}

We first prove two lemmas. With these lemmas, Lemma \ref{corollary of N_0} becomes obvious.

\begin{lemma}\label{Concentrate lemma}\footnote{The logarithm uses base 2.} For any integer $c>2$, $$\sum_{|N_k| \geq c} |N_k| \geq  l\Big({\log} l - \log \big(s_0(f)(s_1(f) - 1)(c - 2) + s_0(f)\big)\Big).$$
\end{lemma}
\begin{proof} First note that for $i\leq l$, $H^i\cap f^{-1}(1)$ is a subcube and co-dim$(H^i\cap f^{-1}(1))=n-m-s_1(f)+1$, which means $|\{k>m|(H^i\cap G)_k\neq \ast\}|=s_1(f)-1$.
W.l.o.g, assume $|N_k|\geq c$ only when $k\in[m+1,m+w]$. For any $y\in\{0,1\}^w$, let $\overline{y}=\{i\in[l]|\forall j\in[w]:(H^i\cap G)_{j+m}\neq y_j\}$.

We claim that for any $y$, $|\overline{y}|$ can not be "too large". Think about the graph $G=(V,E)$ where $V=\overline{y}$ and $(i,j)\in E$ if $i,j\in N_k$ and $(H^i\cap f^{-1}(1))_k\neq(H^j\cap f^{-1}(1))_k$ for some $k>m+w$. It is easy to see that for any $i\in\overline{y}$,
\begin{center}
${\deg}(i)\leq \sum_{k=m+w+1}^{n}[(H^i\cap f^{-1}(1))_k\neq \ast](|N_k|-1)\leq (s_1(f)-1)(c-2)$,
 \end{center}
thus according to Turan's theorem, there exist a independent set $S$ of size $|S|=\frac{|\overline{y}|}{(s_1(f)-1)(c-2)+1}$, which means there exists an input $x\in H$ such that $x^{i}\in f^{-1}(1)$ for any $i\in S$, therefore $|S|\leq s_0(f)$, implying
\begin{equation} \begin{aligned} \label{CL1}
|\overline{y}|\leq ((s_1(f)-1)(c-2)+1)s_0(f).
\end{aligned}\end{equation}
On the other side, let $w_i=|\{k\in[m+1,m+w]|(H^i\cap G)_{k}\neq\ast\}|$, then there are exact $2^{w-w_i}$ $\overline{y}$s containing $i$, thus
\begin{equation} \begin{aligned} \label{CL2}
\sum_{y\in\{0,1\}^n}|\overline{y}|=\sum_{i\leq l}2^{w-w_i}\geq l\cdot 2^{w-\sum_{i\leq l}w_i/l}=l\cdot 2^{w-\sum_{k=m+1}^{m+w}|N_{k}|/l}.
 \end{aligned} \end{equation}
The inequality is due to the AM-GM inequality. From Inequality (\ref{CL1}) and (\ref{CL2}), we can get this lemma.
\end{proof}

\begin{lemma}\label{Balance lemma}  If $l>s_0(f)$, then $\sum_{k > m} \left| |N_k^0| -| N_k^1| \right| \leq l\sqrt{2\ln L(s_1(f) - 1)}.$
\end{lemma}
\begin{proof}
For convenience, assume $|N^0_k|\leq |N^1_k|$ for any $k>m$, and the other cases can follow the same proof. First note that $\sum_{k>m}|N^0_k|\geq 1$, otherwise there exist $x\in H$ such that $x^i\in f^{-1}(1)$ for every $i\in[l]$, which is a contradiction with $l> s_0(f)$. We sample a input $x\in H$  as $\Pr(x_{k}=0)=p$ independently for each $k>m$. Here $p:=\sum_{k>m}|N^0_k|/\sum_{k>m}|N_k|>0$. Recall that for $i\in[l]$, $|\{k>m:(H^i\cap f^{-1}(1))_k\neq\ast\}|=s_1(f)-1$, then $\Pr(x^{i}\in f^{-1}(1))=p^{d_i}(1-p)^{s_1(f)-1-d_i}$, where $d_i:=|\{k>m:(H^i\cap f^{-1}(1))_k=0\}|$. Therefore
\begin{equation}  \label{BL}
s_0(f)\geq  \mathbb{E}(s(f,x))\geq\sum_{i\in[l]}\Pr(x^{i}\in f^{-1}(1))\geq lp^{p(s_1(f)-1)}(1-p)^{(1-p)(s_1(f)-1)}.
 \end{equation}
The last step is due to the AM-GM inequality and the fact that $\sum_{k>m}|N_k|=l(s_1(f)-1)$. By calculus, it is not hard to obtain $e^{2(p-1/2)^2}\leq 2p^p(1-p)^{1-p}$ for $p\leq \frac{1}{2}$. Together with Inequality (\ref{BL}) and recall that $L=s_0(f)2^{s_1(f)-1}/l$, it implies $p\geq\frac{1}{2}(1-\sqrt{\frac{2\ln L}{s_1(f)-1}})$. Therefore
\begin{center}
$\sum_{k>m}(|N_k^1|-|N^0_k|)=(1-2p)\sum_{k>m}|N_k|\leq l\sqrt{2\ln L(s_1(f) - 1)}$.
\end{center}
\end{proof}

Now, Lemma \ref{corollary of N_0} becomes obvious. For any $c_2>2c_1$, first note that
$$\sum_{|N_k^1|<c_1,|N^k|\geq c_2} |N^0_k|-|N^1_k|=\sum_{|N_k^1|<c_1,|N^k|\geq c_2} |N_k|(1-\frac{2|N^1_k|}{|N_k|})\geq \frac{c_2-2c_1}{c_2}\sum_{|N_k^1|<c_1,|N^k|\geq c_2} |N_k|.$$
Then we have
\begin{equation} \begin{aligned} \label{unbalance destination}\nonumber
 \sum_{|N_k^1|\geq c_1} 2|N_k^0| &\geq \sum_{|N_k^1|\geq c_1,|N_k|\geq c_2} 2|N_k^0|\\
 &= \sum_{|N_k| \geq c_2} |N_k|-\sum_{|N_k^1|<c_1,|N_k|\geq c_2} |N_k|-\sum_{|N_k^1|\geq c_1, |N_k|\geq c_2} (|N_k^1|-|N_k^0|)
\\ & \geq \sum_{|N_k|\geq c_2} |N_k|-\frac{c_2}{c_2 - 2c_1}\sum_{|N_k^1|<c_1,|N_k|\geq c_2}( |N_k^0| - |N_k^1|) - \sum_{|N_k^1| \geq c_1, |N_k|\geq c_2} \left| |N_k^1|-|N_k^0|\right|
\\ & \geq \sum_{|N_k|\geq c_2} |N_k|-\frac{c_2}{c_2 - 2c_1}\sum_{|N_k|\geq c_2}\left| |N^0_k|-|N^1_k|\right|.
\end{aligned} \end{equation}
According to Lemma \ref{Concentrate lemma} and Lemma \ref{Balance lemma}, we have

\begin{equation} \begin{aligned} \label{CORO}\nonumber
\sum_{|N_k^1|\geq c_1} |N_k^0| &\geq \frac{l({\log} l - {\log} (s_0(f)(s_1(f) - 1)(c_2 - 1) + s_0(f)))}{2}- \frac{lc_2\sqrt{2\ln L(s_1(f) - 1)}}{2(c_2 - 2c_1)}
\\&=\frac{l(s_1(f)-1-\log L-\log((s_1(f)-1)(c_2-2)+1))}{2}-\frac{lc_2\sqrt{2\ln L(s_1(f)-1)}}{2(c_2-2c_1)}
\end{aligned} \end{equation}
Recall $L\leq 2$, and let $c_1=s_1(f)-1$ and $c_2=3c_1$, thus
$$\sum_{|N_k^1|\geq s_1(f)-1} |N_k^0|\geq l(s_1(f)-1)(\frac{1}{2}-o(1)).$$

\subsection{Proof of Lemma \ref{lemma of heavy cubes}}\label{subsection:Proofs of lemma of heavy cubes}
\begin{proof} Note that $H_{k=1}\cap f^{-1}(1)=\emptyset$ and $H_{k=0}^i\cap f^{-1}(1)=\emptyset$ for $i\in N^1_k$. Thus it is easy to see that
\begin{equation} \begin{aligned} \label{difference of v_1}\nonumber
 s_h(f,H_{k=1}) - s_h(f,H_{k=0}^i)=\sum_{l<t\leq m}\sum_{x\in H^t_{k=1}} (f(x)-f(x^{i,k})).
\end{aligned} \end{equation}
Therefore,
\begin{equation} \begin{aligned} \label{difference of v_2}\nonumber
& \quad \sum_{k:|N_k^1|\geq s_1(f)-1}\sum_{i\in S_k^1} \left(s_h(f,H_{k=1})-s_h(f,H_{k=0}^i)\right)
\\& = \sum_{k:|N_k^1|\geq s_1(f)-1} \sum_{i\in S_k^1} \sum_{l<t\leq m} \sum_{x\in H^t_{k=1}} \left(f(x) - f(x^{i,k})\right)
\\& \leq \sum_{k:|N_k^1|\geq s_1(f)-1} \sum_{i\in S_k^1} \sum_{l<t\leq m} \sum_{\substack{x \in H^t, \\f(x) = 1,f(x^{i,k})=0}} 1
\\& = \sum_{k:|N_k^1|\geq s_1(f)-1} \sum_{i\in S_k^1} \sum_{l<t\leq m} \bigg(\sum_{\substack{x \in H^t, f(x) = 1,\\f(x^{k}) = 0, f(x^{i,k})=0}} 1 +  \sum_{\substack{x\in H^t, f(x) = 1,\\f(x^{k})=1, f(x^{i,k}) = 0}} 1\bigg)
\\& \leq \sum_{l< t \leq m}\bigg(\sum_{x \in H^t, f(x)= 1}\sum_{\substack{k:|N_k^1| \geq s_1(f)-1,\\f(x^{k}) = 0}}\sum_{i\in S_k^1} 1 + \sum_{x\in H^t, f(x^k)= 1}\sum_{\substack{i:f(x^{i,k})=0}}\sum_{k:i\in S_k^1}1 \bigg)
\\& \leq \sum_{l< t \leq m}\bigg(\sum_{x \in H^t, f(x)= 1}\sum_{\substack{k:|N_k^1| \geq s_1(f)-1,\\f(x^{k}) = 0}}(s_1(f)-1) + \sum_{x\in H^t, f(x^k)= 1}\sum_{\substack{i:f(x^{i,k})=0}}(s_1(f)-1)\bigg)
\\& = \sum_{l< t \leq m}\bigg(\sum_{x \in H^t, f(x)= 1}\sum_{\substack{k:|N_k^1| \geq s_1(f)-1,\\f(x^{k}) = 0}}(s_1(f)-1) + \sum_{x\in H^t, f(x)= 1}\sum_{\substack{i:f(x^{i})=0}}(s_1(f)-1)\bigg)
\\& =\sum_{l< t \leq m}\sum_{x \in H^t, f(x)= 1}\bigg(\sum_{k>m,f(x^{k})=0}1+\sum_{i\leq l, f(x^{i})=0}1\bigg)(s_1(f)-1)
\\& \leq\sum_{l< t \leq m}\sum_{x \in H^t, f(x)= 1}(s_1(f)-1)^2= (s_1(f)-1)^2\sum_{l<t\leq m}|H^t\cap f^{-1}(1)|.
\end{aligned} \end{equation}
\end{proof}
%
%
%


\section{Proof of Theorem \ref{theorem of small s_1(f)}}\label{section:Proofs of theorem of small s_1(f)}
\begin{proof}The notation used here is similar to that in section \ref{section:Proofs of the lemmas}. Let $f$ be an $n$-input Boolean function with $s_1(f) = 2$. Let $z$ be the input with $f(z)=0$ and $C(f,z)=C_0 (f)=m$, then there exists a 0-certificate of co-dim $C_0(f)$ consistent with $z$. Let $H$ be the one with maximum average sensitivity if there are many such 0-certificates. Among the $m$ neighbor cubes of $H$, from Lemma \ref{lemma of Ambainis} we have either $|H^i\cap f^{-1}(1)|\geq\frac{3|H|}{4}$ or that $H^i\cap f^{-1}(1)$ is a $1$ co-dimensional subcube of $H^i$ of size $\frac{|H|}{2}$, which are called {\em heavy cube} and {\em light cube} respectively. For light cubes $H^i$ and $H^j$, if $(H^i\cap f^{-1}(1))_k=b$ and $(H^j\cap f^{-1}(1))_k=1 - b$, where $b\in \{0,1\}$, We call $\{H^i, H^j\}$ a pair.  W.l.o.g, assume $H=0^m\ast^{n-m}$ and there are $l$ light cubes and $l_1/2$  mutually disjoint pairs. Moreover, assume that the $l_1$ cubes in pairs are $H^{1},H^{2}, \dots, H^{l_1}$, the $l_2:=l-l_1$ other light cubes are $H^{l_1+1},H^{l_1+2}, \dots, H^{l}$ and the $h$ heavy cubes are $H^{l+1},H^{l+2}, \dots, H^ {m}$. In addition, assume $\{i\in [l_1, l]|(H^i\cap f^{-1}(1))_k=1\} = 1$ for $k>m$ by flipping the bits.

The main idea is to prove two inequalities: $s_0(f)\geq \frac{5l_1}{8}+\frac{l_2+h}{2}$ and $s_0(f)\geq \frac{l_1}{2}+(1-p)l_2+2ph-p^2h$ for any $0\leq p\leq \frac{1}{2}$, with which we would obtain the conclusion through a little calculation.



 The first inequality is due to the following lemma~\cite{APV16}.
\begin{lemma}\label{pair lemma} ~\cite{APV16} Let $P$ be a set of mutually disjoint pairs of the neighbour cubes of $H$. Then there exist a $0$-certificate $H'$ such that $z \in H'$, $dim(H) = dim(H')$ and $H'$ has at least $|P|$ heavy neighbour cubes.
\end{lemma}
\noindent Thus
\begin{equation} \begin{aligned} \label{average sensitivity of Q_n_transorm}
s(f) \geq \frac{s(f,H')}{|H'|} \geq \frac{1}{2}(l_1 + l_2 + h - \frac{l_1}{2}) + \frac{3}{4} \times \frac{1}{2}l_1 = \frac{5l_1}{8} + \frac{l_2 + h}{2}.
\end{aligned} \end{equation}

We show the second inequality by the probabilistic method.
 We sample a input $x\in H$ as $\Pr(x_{k}=0)=p$ independently for each $k>m$. Recall that for $i\in[l]$, $|\{k>m:(H^i\cap f^{-1}(1))_k\neq\ast\}|=1$, then $\Pr(x^{i}\in f^{-1}(1))=p^{d_i}(1-p)^{1-d_i}$, where $d_i:=|\{k>m:(H^i\cap f^{-1}(1))_k=0\}|$. Therefore
 \begin{equation} \begin{aligned} \label{BLS1(f)2}
\sum_{i\in[l]}\Pr(x^{i}\in f^{-1}(1)) & = \sum_{i \leq l_1}p^{d_i}(1-p)^{1-d_i} + \sum_{l_1< i \leq l_1 + l_2}p^{d_i}(1-p)^{1-d_i}
\\ & = \frac{l_1}{2}(p + 1 - p) + l_2(1 - p)=\frac{l_1}{2}+(1-p)l_2.
 \end{aligned}\end{equation}
For any heavy cube $H^i$, where $l<i\leq m$, we claim that $\Pr(f(x^i)=1)\geq 2p-p^2$, which implies the inequality since
$$s_0(f)\geq \mathbb{E}(s(f,x))=\sum_{i=1}^{m}\Pr(f(x^i)=1)\geq \frac{l_1}{2}+(1-p)l_2+2ph-p^2h.$$
\par Let $C\subseteq H^i\cap f^{-1}(0)$ be a maximal 0-certificate and $N(C)$ be the set of vertices adjacent to $C$. Here we say a certificate is maximal if it is not contained in a larger one. Then according to the assumption that $s_1(f)=2$, it is easy to see $f(x)=1$ for any $x\in N(C)$. Thus $H^i\cap f^{-1}(0)$  can be decomposed into disjoint maximal 0-certificates, denoted by $\{C_1, C_2,\cdots\}$. Moreover, we also have $N(C_{j_1})\cap N(C_{j_2})=\emptyset$ if $j_1\neq j_2$, since $s(y)\geq 3$ for $y\in N(C_{j_1})\cap N(C_{j_2})$. For each $C$, let $D= |\{k>m:C(k)\neq\ast\}|$ and $D_0 = |\{k>m:C(k) = 0\}|$.
Note that
\begin{center}
$\Pr(x^{i}\in C)  = p^{D_0}(1-p)^{D-D_0}.$
\end{center}
If $D \leq 2$, from Lemma \ref{lemma of Ambainis} we have $H^i \cap f^{-1}(0) = C$. Therefore,
\begin{equation} \begin{aligned}\label{Ratio_Dleq2}
\Pr(x^{i}\in f^{-1}(1))= 1 - \Pr(x^{i}\in C)
= 1 - p^{D_0}(1-p)^{D-D_0}
\geq 1- (1-p)^2.
\end{aligned}\end{equation}
If $D\geq 3$, it is not hard to see
\begin{equation} \begin{aligned}\label{Ratio_NC}
\Pr(x^{i}\in N(C))  & =  D_0p^{D_0 - 1}(1-p)^{D-D_0 + 1} + (D-D_0)p^{D_0 + 1}(1-p)^{D-D_0 - 1}
\\ & \geq D_0p^{D_0 + 1}(1-p)^{D-D_0 - 1} + (D-D_0)p^{D_0 + 1}(1-p)^{D-D_0 - 1}
\\ & = \bigg(\frac{Dp}{1 - p}\bigg)\Pr(x^{i}\in C) \geq \bigg(\frac{3p}{1 - p}\bigg)\Pr(x^{i}\in C)
\\ &\geq \bigg(\frac{1}{(1 - p)^2}-1\bigg)\Pr(x^{i}\in C).
\end{aligned}\end{equation}

%
Thus
\begin{equation} \begin{aligned}\label{Ratio_Dgeq3}\nonumber
\Pr(f(x^i)=1)&\geq \sum_{t}\Pr(x^i\in N(C_t))\geq \bigg(\frac{1}{(1 - p)^2}-1\bigg)\sum_{t}\Pr(x^i\in C_t)
\\&=\bigg(\frac{1}{(1 - p)^2}-1\bigg)\Pr(f(x^i)=0).
\end{aligned}\end{equation}
Therefore, $\Pr(x^{i}\in f^{-1}(1))\geq 1 - (1 - p)^2 = 2p-p^2$.
\par Now we have shown the two inequalities, that is,
\begin{equation}  \label{Result}
s_0(f)\geq \max \left\{ \max_{0 \leq p \leq \frac {1}{2}}\{l_1/2 + l_2(1 - p) + 2ph -  p^2h\} ,
 \frac{5}{8}l_1 + \frac{l_2 + h}{2} \right\}.
 \end{equation}
If $h \leq l_2 \leq 2h$, let $p = 1 - \frac {l_2}{2h}$, we have
\begin{equation} \begin{aligned} \label{Result Case 1_1}\nonumber
s_0(f)\geq \max \big\{\frac {l_1}{2} + \frac {l_2^2}{4h} + h , \frac{5l_1}{8} + \frac{l_2 + h}{2}\big\}.
\end{aligned} \end{equation}
Let $l'_2 = \frac {l_2}{l_1 + l_2 + h}$, $h' = \frac {h}{l_1 + l_2 + h}$, we get
\begin{equation} \begin{aligned} \label{Result Case 1_2}\nonumber
s_0(f) &\geq (l_1 + l_2 + h)\max \left\{\frac {1}{2} - \frac {l'_2}{2} + \frac {h'}{2} + \frac {{l'_2}^2}{4h'}, \frac{5}{8} - \frac{l'_2 + h'}{8} \right\}.
\end{aligned} \end{equation}
Let $g_1(l'_2,h')=\frac {1}{2} - \frac {l'_2}{2} + \frac {h'}{2} + \frac {{l'_2}^2}{4h'}$, $g_2(l'_2,h') = \frac{5}{8} - \frac{l'_2 + h'}{8}$ and $x(h') = \frac {3h' + \sqrt {8h' - 31h^{'2}}}{4}$. We have $g_1(x(h'),h') = g_2(x(h'),h')$ and $\max\{g_1(l'_2,h'),g_2(l'_2,h')\} \geq g_1(x(h'),h')$, because $g_1(l'_2,h')$ is monotone increasing and $g_2(l'_2,h')$ is monotone decreasing if $l'_2$ increases. By calculating the zero point of the derivative of function $g_1(x(h'),h')$, we have $g_1(x(h'),h')$ takes the minimum value at $h'_* = \frac {20 + 7\sqrt 5} {155} $. Therefore,
\begin{equation} \begin{aligned} \label{Result Case 1_1}
s_0(f) &\geq (l_1 + l_2 + h)\max\{g_1(l'_2,h'),g_2(l'_2,h')\}
\\&\geq (l_1 + l_2 + h)g_1(x(h'),h')
\\&\geq (l_1 + l_2 + h)g_1(x(h'_*),h'_*)
\\&= \frac {(37 - \sqrt 5)(l_1 + l_2 + h)}{62}.
\end{aligned} \end{equation}

If $l_2 \leq h$, let $p = \frac {1}{2}$. From Inequality (\ref{Result}) we have
\begin{equation} \begin{aligned} \label{Result Case 2_1}
s_0(f) &\geq \max \big\{\frac {l_1 + l_2}{2} + \frac {3h}{4} , \frac{5l_1}{8} + \frac{l_2 + h}{2} \big\}
\\ &\geq \max \big\{\frac {l_1}{2}+ \frac {5(l_2 + h)}{8}, \frac{5l_1}{8} + \frac{l_2 + h}{2} \big\}
\\ &\geq \frac{1}{2}\bigg(\frac {l_1}{2}+ \frac {5(l_2 + h)}{8} + \frac{5l_1}{8} + \frac{l_2 + h}{2}\bigg)
\\& = \frac {9(l_1 + l_2 + h)}{16}.
\end{aligned} \end{equation}
The second inequality is due to $l_2 \leq h$.
\par If $l_2 \geq 2h$, let $p = 0$. From Inequality (\ref{Result}) we have
\begin{equation} \begin{aligned} \label{Result Case 3_1}
s_0(f) &\geq \max \big\{\frac {l_1}{2} + l_2, \frac{5l_1}{8} + \frac{l_2 + h}{2} \big\}
\\ &\geq \max \big\{\frac {l_1}{2} + \frac {2(l_2 + h)}{3}, \frac{5l_1}{8} + \frac{l_2 + h}{2}  \big\}
\\ &\geq \frac{3}{7}\bigg(\frac {l_1}{2} + \frac {2(l_2 + h)}{3}\bigg) + \frac{4}{7}\bigg(\frac{5l_1}{8} + \frac{l_2 + h}{2}  \bigg)
\\ & = \frac {4(l_1 + l_2 + h)}{7}.
\end{aligned} \end{equation}
The second inequality is due to $l_2 \geq 2h$.

Combining inequality (\ref{Result Case 1_1}), (\ref{Result Case 2_1}) and (\ref{Result Case 3_1}), we have
\begin{equation} \begin{aligned} \label{Conclusion}\nonumber
s_0(f) \geq \frac {(37 - \sqrt 5)(l_1 + l_2 + h)}{62}.
\end{aligned} \end{equation}
Therefore,
\begin{equation} \begin{aligned} \label{Theorem}\nonumber
c_0(f) = l_1 + l_2 + h \leq \frac {37 + \sqrt 5}{22}s_0(f).
\end{aligned} \end{equation}
\end{proof}

\section{Conclusion}\label{section:conclusion}
In this work, we give a better upper bound of block sensitivity in terms of sensitivity. Our results are based on carefully exploiting the structure of the light cubes. However, our approach has an obvious limitation. In the extremal case, if there are no light cubes, then we can only get $bs(f)\leq C(f)\leq\frac{2}{3}s(f)2^{s(f)-1}$.  Better understanding about the structure of heavy cubes is needed in order to conquer this limitation.

\bibliographystyle{plain}
\bibliography{reference}

\begin{thebibliography}{10}

\bibitem{ABG+14}
Andris Ambainis, Mohammad Bavarian, Yihan Gao, Jieming Mao, Xiaoming Sun, and
  Song Zuo.
\newblock Tighter relations between sensitivity and other complexity measures.
\newblock In {\em Automata, Languages, and Programming - 41st International
  Colloquium, {ICALP} 2014, Copenhagen, Denmark, July 8-11, 2014, Proceedings,
  Part {I}}, pages 101--113, 2014.

\bibitem{AP14}
Andris Ambainis and Krisjanis Prusis.
\newblock A tight lower bound on certificate complexity in terms of block
  sensitivity and sensitivity.
\newblock In {\em Mathematical Foundations of Computer Science 2014 - 39th
  International Symposium, {MFCS} 2014, Budapest, Hungary, August 25-29, 2014.
  Proceedings, Part {II}}, pages 33--44, 2014.

\bibitem{APV16}
Andris Ambainis, Krisjanis Prusis, and Jevgenijs Vihrovs.
\newblock Sensitivity versus certificate complexity of boolean functions.
\newblock In {\em Computer Science - Theory and Applications - 11th
  International Computer Science Symposium in Russia, {CSR} 2016, St.
  Petersburg, Russia, June 9-13, 2016, Proceedings}, pages 16--28, 2016.

\bibitem{AS11}
Andris Ambainis and Xiaoming Sun.
\newblock New separation between $s(f)$ and $bs(f)$.
\newblock {\em CoRR}, abs/1108.3494, 2011.

\bibitem{AV15}
Andris Ambainis and Jevgenijs Vihrovs.
\newblock Size of sets with small sensitivity: {A} generalization of simon's
  lemma.
\newblock In {\em Theory and Applications of Models of Computation - 12th
  Annual Conference, {TAMC} 2015, Singapore, May 18-20, 2015, Proceedings},
  pages 122--133, 2015.

\bibitem{MFCS16}
Mitali Bafna, Satyanarayana~V. Lokam, S{\'{e}}bastien Tavenas, and Ameya
  Velingker.
\newblock On the sensitivity conjecture for read-k formulas.
\newblock In {\em 41st International Symposium on Mathematical Foundations of
  Computer Science, {MFCS} 2016, August 22-26, 2016 - Krak{\'{o}}w, Poland},
  pages 16:1--16:14, 2016.

\bibitem{Sha16}
Shalev Ben{-}David.
\newblock Low-sensitivity functions from unambiguous certificates.
\newblock {\em CoRR}, abs/1605.07084, 2016.

\bibitem{Bop12}
Meena Boppana.
\newblock Lattice variant of the sensitivity conjecture.
\newblock {\em Electronic Colloquium on Computational Complexity {(ECCC)}},
  19:89, 2012.

\bibitem{buhrman}
Harry Buhrman and Ronald De~Wolf.
\newblock Complexity measures and decision tree complexity: a survey.
\newblock {\em Theoretical Computer Science}, 288(1):21--43, 2002.

\bibitem{Chakraborty}
Sourav Chakraborty.
\newblock On the sensitivity of cyclically-invariant boolean functions.
\newblock In {\em Proceedings of the 20th Annual IEEE Conference on
  Computational Complexity}, CCC '05, pages 163--167, Washington, DC, USA,
  2005. IEEE Computer Society.

\bibitem{Cook}
Stephen Cook and Cynthia Dwork.
\newblock Bounds on the time for parallel ram's to compute simple functions.
\newblock In {\em Proceedings of the Fourteenth Annual ACM Symposium on Theory
  of Computing}, STOC '82, pages 231--233, New York, NY, USA, 1982. ACM.

\bibitem{CDR86}
Stephen~A. Cook, Cynthia Dwork, and R{\"{u}}diger Reischuk.
\newblock Upper and lower time bounds for parallel random access machines
  without simultaneous writes.
\newblock {\em {SIAM} J. Comput.}, 15(1):87--97, 1986.

\bibitem{GKS15}
Justin Gilmer, Michal Kouck{\'{y}}, and Michael~E. Saks.
\newblock A new approach to the sensitivity conjecture.
\newblock In {\em Proceedings of the 2015 Conference on Innovations in
  Theoretical Computer Science, {ITCS} 2015, Rehovot, Israel, January 11-13,
  2015}, pages 247--254, 2015.

\bibitem{GNS+16}
Parikshit Gopalan, Noam Nisan, Rocco~A. Servedio, Kunal Talwar, and Avi
  Wigderson.
\newblock Smooth boolean functions are easy: Efficient algorithms for
  low-sensitivity functions.
\newblock In {\em Proceedings of the 2016 {ACM} Conference on Innovations in
  Theoretical Computer Science, Cambridge, MA, USA, January 14-16, 2016}, pages
  59--70, 2016.

\bibitem{GSTW16}
Parikshit Gopalan, Rocco~A. Servedio, Avishay Tal, and Avi Wigderson.
\newblock Degree and sensitivity: tails of two distributions.
\newblock {\em Electronic Colloquium on Computational Complexity {(ECCC)}},
  23:69, 2016.

\bibitem{HKP11}
Pooya Hatami, Raghav Kulkarni, and Denis Pankratov.
\newblock Variations on the sensitivity conjecture.
\newblock {\em Theory of Computing, Graduate Surveys}, 4:1--27, 2011.

\bibitem{KK04}
Claire Kenyon and Samuel Kutin.
\newblock Sensitivity, block sensitivity, and l-block sensitivity of boolean
  functions.
\newblock {\em Inf. Comput.}, 189(1):43--53, 2004.

\bibitem{LZ16}
Chengyu Lin and Shengyu Zhang.
\newblock Sensitivity conjecture and log-rank conjecture for functions with
  small alternating numbers.
\newblock {\em CoRR}, abs/1602.06627, 2016.

\bibitem{nisan1991crew}
Noam Nisan.
\newblock Crew prams and decision trees.
\newblock {\em SIAM Journal on Computing}, 20(6):999--1007, 1991.

\bibitem{NS94}
Noam Nisan and Mario Szegedy.
\newblock On the degree of \makeatletter{B}oolean functions as real
  polynomials.
\newblock {\em Computational Complexity}, 4:301--313, 1994.

\bibitem{ST16}
Karthik~C. S. and S{\'{e}}bastien Tavenas.
\newblock On the sensitivity conjecture for disjunctive normal forms.
\newblock {\em CoRR}, abs/1607.05189, 2016.

\bibitem{Sze15}
Mario Szegedy.
\newblock An $o(n^{0.4732})$ upper bound on the complexity of the {GKS}
  communication game.
\newblock {\em CoRR}, abs/1506.06456, 2015.

\bibitem{Tal16}
Avishay Tal.
\newblock On the sensitivity conjecture.
\newblock In {\em 43rd International Colloquium on Automata, Languages, and
  Programming, {ICALP} 2016, July 11-15, 2016, Rome, Italy}, pages 38:1--38:13,
  2016.

\bibitem{Turan}
Gy{\"o}rgy Tur{\'a}n.
\newblock The critical complexity of graph properties.
\newblock {\em Information Processing Letters}, 18(3):151--153, 1984.

\end{thebibliography}

\end{document}